\documentclass{article}

\usepackage[preprint]{neurips_2022}

\usepackage[utf8]{inputenc} 
\usepackage[T1]{fontenc}    
\usepackage{url}            
\usepackage{booktabs}       
\usepackage{amsfonts,amsmath,amsthm, amssymb}       
\usepackage{nicefrac}       
\usepackage{microtype}      
\usepackage{xcolor}         
\usepackage{bm}
\newtheorem{theorem}{Theorem}[section]
\newtheorem{definition}{Definition}[section]
\usepackage[hidelinks]{hyperref}       
\usepackage{multirow}
\usepackage[scr=boondoxo,frak=euler,bb=ams]{mathalfa}
\usepackage{graphicx,graphics} 
\usepackage{enumerate,enumitem} 
\usepackage[linesnumbered,ruled,lined]{algorithm2e} 
\SetKwInput{KwInput}{Input}      
\SetKwInput{KwOutput}{Output}    

\newcommand{\norm}[1]{\left\lVert#1\right\rVert}
\newcommand{\abs}[1]{\lvert#1\rvert} 
\newcommand{\rev}[1]{#1}

\title{Introducing the Huber mechanism for differentially private low-rank matrix completion}

\author{%
	R Adithya Gowtham* \quad Gokularam M* \quad Thulasi Tholeti* \quad Sheetal Kalyani\\
	Department of Electrical Engineering\\
	IIT Madras\\
	\texttt{\{ee17b146@smail,ee17d400@smail,ee15d410@ee,skalyani@ee\}.iitm.ac.in}\\
}

\begin{document}
	
	\maketitle
	
	\begin{abstract}
		Performing low-rank matrix completion with sensitive user data calls for privacy-preserving approaches. In this work, we propose a novel noise addition mechanism for preserving differential privacy where the noise distribution is inspired by Huber loss, a well-known loss function in robust statistics. The proposed Huber mechanism is evaluated against existing differential privacy mechanisms while solving the matrix completion problem using the Alternating Least Squares approach. We also propose using the Iteratively Re-Weighted Least Squares algorithm to complete low-rank matrices and study the performance of different noise mechanisms in both synthetic and real datasets. \rev{We prove that the proposed mechanism achieves $\epsilon$-differential privacy similar to the Laplace mechanism. Furthermore, empirical results indicate that the Huber mechanism outperforms Laplacian and Gaussian in some cases and is comparable, otherwise.} 
	\end{abstract}
	
	
	\section{Introduction}
	\footnotetext{*Equal contribution}
	Recovering a low-rank matrix from a set of limited observations is an important problem in machine learning and data analysis. It finds applications in multiple areas such as recommender systems \cite{luo2014efficient}, image restoration \cite{he2015total}, and phase retrieval (For a full survey, see \cite{nguyen2019low}). Its use in recommender systems was popularized by \cite{netflix} where the user fills in surveys regarding a small fraction of the items viewed and the system is expected to provide a recommendation by estimating the remaining entries of the matrix. 
	
	Many algorithms have been proposed in literature to solve this problem which typically involves matrix decomposition \cite{lu2015nonconvex,liu2013ftf}. In \cite{jain2012altmin}, the target matrix $\mathbf{X}$ is formulated as a bi-linear problem $\mathbf{X}=\mathbf{U}\mathbf{V}^\top$; the authors propose alternating minimization over the matrices $\mathbf{U}$ and $\mathbf{V}$. 
	The nuclear norm of the matrix is minimized in \cite{cai2008svt} in order to recover the matrix with the lowest rank. Singular value decomposition is performed on the target matrix, and a threshold is applied iteratively over the singular values to converge to a proven unique solution. Most of these solution methods involve iterative procedures as a part of the optimization algorithm that converges to a desired solution.
	
	Given that the user data, i.e., limited observations, available to the matrix completion algorithms are highly sensitive, there has been growing interest in providing privacy guarantees while handling such data.  Differential privacy was introduced by \cite{dwork2006privacy} as a method of preserving sensitive data about an individual while providing statistical information regarding the dataset as a whole. The framework has been widely adapted to provide private versions of well-established algorithms in varied fundamental areas such as clustering, gradient descent, localization, deep learning, data mining and many more (\cite{abadi2016deep,xia2020distributed,song2013stochastic,friedman2010data}).
	
	To that end, many differentially private low-rank matrix completion (LRMC) algorithms have been proposed (\cite{liu2015,jain2018differentially}) where provable privacy (and sometimes, performance) guarantees are provided. The differentially private version of alternating least squares is proposed in \cite{chien2021private} where noise is introduced in the optimization procedure to provide privacy guarantees. The algorithms mentioned above use a combination of trimming techniques (clamping entries of high magnitude) and Gaussian noise addition to provide an $(\epsilon,\delta)$-differential privacy guarantee which is a relaxed notion of privacy. 
	
	In this work, we explore the choice of using other noise addition mechanisms to achieve differentially private LRMC. We propose a novel addition mechanism, the Huber mechanism, that combines the advantages of the Gaussian and Laplace mechanisms and provide privacy guarantees for the proposed framework. We also implement the Alternating Least Squares algorithm with different noise addition mechanisms and compare their performance. In addition, we also propose a new optimization method inspired by the best way to post-process the noise added by the proposed mechanism so as to achieve accuracy without compromising on privacy. The contributions of our work are detailed as follows:
	\begin{itemize}
	    \item We introduce a differential privacy mechanism, the Huber mechanism, which adds noise from the Huber distribution, which is a combination of Laplace and Gaussian distributions. This mechanism is proposed to achieve the advantages of both Laplacian and Gaussian mechanisms. We prove that the proposed mechanism achieves $\epsilon$-differential privacy.
	    \item We propose two variations of differentially private low-rank matrix completion algorithms that use Alternating Least Squares (ALS) and Iteratively Re-weighted Least Squares (IRLS), where Huber noise is added to achieve privacy.
	    \item We provide extensive simulation results on synthetic as well as real datasets comparing the performance of the Huber mechanism to other standard noise addition mechanisms for differentially private LRMC.
	\end{itemize}
	
	\paragraph{Notation}
	Let $\mathbb{N}_p^{}$ be the set of first $p$ natural numbers, $\mathbb{N}_p^{}=\{1,\,2,\,\ldots,\,p\}$. $\mathcal{N}(\mu,\sigma^2_{})$ denotes the normal distribution with mean $\mu$ and variance $\sigma^2_{}$ and $\mathcal{L}(\beta)$ denotes the Laplace distribution with scale parameter $\beta$. Let $\Phi
	(\cdot)$ denote the 
	distribution function of the standard normal distribution. Throughout the article, matrices are denoted by uppercase bold letters and vectors are denoted by lowercase bold letters. 
	The $(i,\,j)$-th entry, the $i$-th row  (as a column vector)  and the $j$-th column of the matrix $\mathbf{B}\in \mathbb{R}^{p\times q}_{}$ are respectively denoted as $b_{ij}^{},\ \mathbf{b}_i^{}$ and $\widetilde{\mathbf{b}}_j^{},$ where $i\in\mathbb{N}_p^{}$ and $j\in\mathbb{N}_q^{}$.
	$\mathbf{B}_{\mathcal{I},\,\mathcal{J}}^{}$ denotes the sub-matrix of $\mathbf{B},$ formed with the entries $b_{ij}^{},\ i\in\mathcal{I}$ and $j\in\mathcal{J},$ where $\mathcal{I}$ and $\mathcal{J}$ are the index sets. 
	Similarly, for a vector $\mathbf{v}\in \mathbb{R}^{p}_{},\ v_{i}^{}$ denotes its $i$-th element, $i\in\mathbb{N}_p^{}$ and $\mathbf{v}_{\mathcal{I}}^{}$ denotes the sub-vector formed with the entries $v_{i}^{},\ i\in\mathcal{I}$. 
	Let $\operatorname{{Diag}}(\zeta_1^{},\,\zeta_2^{},\,\ldots,\,\zeta_p^{})$ denote the $p\times p$ diagonal matrix with the diagonal entries $\zeta_1^{},\,\zeta_2^{},\,\ldots,\,\zeta_p^{}$. The $\ell_p^{}$-norm of a vector $\mathbf{b}$ is denoted as $\norm{\mathbf{b}}_p^{}$ and the nuclear and Frobenius norms of the matrix $\mathbf{B}$ are denoted by $\norm{\mathbf{B}}_{*}^{}$ and $\norm{\mathbf{B}}_{\texttt{F}}^{}$ respectively. 
	
	
	\section{Differential privacy}
	
	\subsection{Background}
	Differential privacy aims to preserve the privacy of an individual while allowing meaningful inferences from the entire dataset. A privacy mechanism $\mathcal{M}$ is an algorithm that takes a data matrix as input and provides an output to a query. Ideally, the output provides accurate responses to the queries without compromising on individual data. Differential privacy is formally defined below and is reproduced from \cite{dwork2014algorithmic} for ease of reading.
	\begin{definition}
		A randomized mechanism $\mathcal{M}$ is said to preserve $\epsilon$-differential privacy if for all datasets $\mathcal{D},\ \mathcal{D}'_{} \in \mathcal{Z}$ that differ on a single element and for all possible sets $\mathcal{S},$
		\begin{equation*} \label{eqn:epriv}
		\operatorname{Pr}(\mathcal{M}(\mathcal{D})  \in \mathcal{S}) \leq \exp(\epsilon)\,  \operatorname{Pr}(\mathcal{M}(\mathcal{D}'_{})  \in \mathcal{S}).
		\end{equation*}
	\end{definition}
	A weaker notation of privacy is defined for cases when $\epsilon$ differential privacy is obtained for a major probability excluding a small fraction $\delta$. 
	\begin{definition}
		A randomized mechanism $\mathcal{M}$ is said to preserve $(\epsilon,\delta)$- differential privacy if for all datasets $\mathcal{D}$ and $\mathcal{D}'_{}$ that differ on a single element and for all possible sets $\mathcal{S},$
		\begin{equation*} \label{eqn:eprivd}
		\operatorname{Pr}(\mathcal{M}(\mathcal{D})  \in \mathcal{S}) \leq \exp(\epsilon)\, \operatorname{Pr}(\mathcal{M}(\mathcal{D}'_{})  \in \mathcal{S}) + \delta.
		\end{equation*}
	\end{definition}
	Sensitivity plays an important role in quantifying differential privacy. It measures the magnitude of change an output $f({\mathcal{D}})$ incurs based on the change of a single data point in the worst-case scenario. The amount of noise required to preserve privacy depends on the sensitivity. It is often defined with respect to a specific norm. Here, we formally define sensitivity with a general norm.
	
	\begin{definition}
	The $\ell_p^{}$-sensitivity of a function the function $f : \mathcal{Z}\rightarrow \mathbb{R}^K_{}$ is given by
		\begin{equation}
		\Delta f_p^{}=\max_{\mathcal{D},\,\mathcal{D}'_{}}
		\, \norm{f(\mathcal{D})-f(\mathcal{D}'_{})}_p^{},
	\end{equation}
	where $\mathcal{D,}\in\mathcal{Z} \text{ and } \mathcal{D}'_{}\in\mathcal{Z}$ are neighbouring datasets differing only by a single data entry.
	\end{definition}
	
	Differential privacy is typically achieved using the Gaussian or Laplacian mechanisms. As the names suggest, they involve the addition of Gaussian or Laplacian noise respectively. 
	Let $\mathbf{t}\in \mathbb{R}^K$ be the noise that is added to $f(\mathcal{D})$ to ensure privacy  i.e,
	\begin{equation*}
	    \mathcal{M}(\mathcal{D}, f(\cdot))=f(\mathcal{D})+\mathbf{t}.
	\end{equation*}
	In the Laplace mechanism, the coordinates of $\mathbf{t}$ are i.i.d samples from $\mathcal{L}(\beta)$ and it is shown to provide $\epsilon$-DP with $\epsilon=\Delta f_1/\beta$. Similarly in the Gaussian mechanism, the entries of $\mathbf{t}$ are i.i.d samples from $\mathcal{N}(0,\sigma^2_{})$ and it offers $(\epsilon,\delta)$-DP with $\epsilon=\sqrt{\frac{2}{\sigma^{2}}\log(1.25/\delta)}\,\Delta f_2$. 
	The choice of the privacy mechanism depends on the application and its specific requirements. Note that there is always a trade-off between accuracy and privacy. Higher privacy is achieved by adding noise of a larger magnitude and this, in turn, will affect the accuracy of the outcome. Although the Laplace mechanism offers a higher degree of privacy, the Gaussian mechanism is often preferred for machine learning applications due to its various advantages. The primary advantage is that for applications in which $\ell_2^{}$-sensitivity is much lower than $\ell_1^{}$-sensitivity in higher dimensions, the Gaussian mechanism allows adding much less noise. We now propose a new mechanism termed the Huber mechanism which aims to combine the advantages of Gaussian and Laplace mechanisms and employ it to perform differentially private low-rank matrix completion in the further sections.
	
	\subsection{Introducing the Huber mechanism}
	
	In this work, we introduce a new mechanism for differential privacy that combines the advantages of the Laplacian and Gaussian mechanisms. The noise distribution is inspired by the Huber loss function proposed in \cite{huber1992robust} and is defined as follows
	\begin{equation}\label{eqn:hub_pdf}
	p(t;\,\alpha) = \kappa_{\alpha}^{} \exp\left({-{\rho}_{\alpha}^{}(t)}\right), 
	\end{equation}
	where $\kappa_{\alpha}^{}$ is the normalizing constant given by 
	$\kappa_{\alpha}^{}
	=\left(\frac{2}{\alpha}\exp\left(\frac{-\alpha^2_{}}{2}\right) 
	+\sqrt{2\pi}
	\left(2\Phi(\alpha)-1\right)
	\right)^{-1}_{}$ and
	\begin{equation}\label{eqn:huber_loss}
	{\rho}_{\alpha}^{}(t) = \begin{cases}
	{t^2_{}}/{2}  &,\  \abs{t} \leq \alpha \\[2pt]
	\alpha\left(\abs{t} - {\alpha}/{2}\right)  &,\  \abs{t} > \alpha
	\end{cases} 
	\end{equation} 
	denotes the \textit{Huber loss function} with the transition parameter $\alpha$. 
	The derivative of the Huber loss function is known as the \textit{Huber influence function}, which is given by
	\begin{equation}\label{eqn:huber_inf}
	\psi_{\alpha}^{}(t) = \operatorname{sign}(t)\cdot\min(|t|,\,\alpha)=
	\begin{cases}
	-\alpha & ,\ \ t < -\alpha \\[2pt]
	t & ,\ \ -\alpha \leq t \leq \alpha\\[2pt]
	\alpha & ,\ \ t > \alpha
	\end{cases}.
	\end{equation}
	Huber distribution is symmetric with exponential tails and a Gaussian center. The parameter $\alpha$ offers flexibility to achieve the desired combination of Laplacian and Gaussian distributions. Given this unique property of the Huber distribution, we propose the Huber mechanism which adds Huber noise to assure privacy.
	
	\begin{definition}[Huber mechanism]
	Given a function $f :\mathcal{Z}\rightarrow\mathbb{R}^k_{}$, the Huber mechanism is defined as
	\begin{equation}
	\mathcal{M}_H^{}(\mathcal{D},\, f(\cdot)) = f(\mathcal{D}) + \mathbf{t},
	\end{equation}
	where $\mathbf{t}$ is a vector of $K$ i.i.d random variables distributed as \eqref{eqn:hub_pdf}. 
\end{definition}
Because of exponential tails of Huber distribution, the Huber mechanism achieves the stronger $\epsilon$-privacy, similar to the Laplacian. As the center of the distribution is like the Gaussian, for higher values of $\alpha,$ the noise samples will be similar 
to Gaussian noise. This, in turn, \rev{suggests} that composite noise addition using the Huber mechanism also allows the addition of much lesser noise when compared to the Laplacian. Next, we formally derive the privacy guarantees of the Huber mechanism.\\
	
	
	\begin{theorem}
		The Huber mechanism guarantees $\epsilon$-differential privacy with $\epsilon=\alpha\cdot\Delta f_1^{},$ where $\Delta f_1^{}$ is the $\ell_1$-sensitivity of the model.
	\end{theorem}
	\begin{proof}
		Given that $\Delta f_1^{}$ is the sensitivity of the model, we determine the ratio of Huber probabilities between a given point $t$ and $t+\Delta f_1^{}$. As $\Delta f_1^{}$ indicates the maximum distance between two data points in a model, this ratio is indicative of $\epsilon$ according to Definition \ref{eqn:epriv}. Note that this is also similar to the ratio considered to provide differential privacy guarantees for the Laplace mechanism in \cite{dwork2014algorithmic}.
		\begin{align*}
		\frac{p(t | \alpha)}{p(t + \Delta f_1^{} | \alpha)} & = \frac{\kappa_{\alpha}^{}\exp({-{\rho}_{\alpha}^{}(t)})} {\kappa_{\alpha}^{}\exp({-{\rho}_{\alpha}^{}(t + \Delta f_1^{})})}
	    = \exp({{\rho}_{\alpha}^{} (t + \Delta f_1^{}) - {\rho}_{\alpha}^{}(t)}) 
		\end{align*}
		Define $g(t)$ as
		$g(t) = {\rho}_{\alpha}^{} (t + \Delta f_1^{}) - {\rho}_{\alpha}^{}(t). $
		
		It is evident from the nature of the two functions ${\rho}_{\alpha}^{}(t + \Delta f_1^{})$ and ${\rho}_{\alpha}^{}(t)$ that $g(t)$ is piece-wise defined. We consider two cases based on the behaviour of the function: $\Delta f_1^{} \leq 2\alpha$ and  $\Delta f_1^{} > 2\alpha$. We derive the bounds for the two cases separately in Appendix \ref{sec:A1} and verify that the upper bound for $g(t)$ in both cases is given by $\alpha\cdot\Delta f_1^{}$.
	\end{proof}
	
	\paragraph{Privacy Analysis} 
	To compare the privacy obtained by using the various noise addition mechanisms, we take a look at the privacy budget of the mechanisms for noise of similar variance. We also assume that the sensitivity of all algorithms $\Delta f_1 = \Delta f_2 = \Delta f$, where we set $\Delta f$ to the maximum possible value, the range of all the datasets considered in the experiments section. Hence, we set $\Delta f = 5$ and tabulate the budget required for similar variance noise addition in Table ~\ref{tab:privacy}:
	\begin{table}[!htb]
		\centering
		\caption{Privacy budgets of the various mechanisms for set noise variance.}
		\begin{tabular}{| c | c | c | c |}\hline
			\textbf{Noise Variance}
			&  \textbf{Gaussian Mechanism} & \textbf{Laplacian Mechanism}
			& \textbf{Huber Mechanism}
			\\ \hline
			 $1$ & $(15.964, 10^{-5})$ & $(7.071, 0)$ & $(15.000, 0)$ \\
			 $2$ & $(11.288, 10^{-5})$ & $(5.000, 0)$ & $(5.382, 0)$ \\
			 $3$ & $(9.217, 10^{-5})$ & $(4.082, 0)$ & $(4.235, 0)$ \\
			 $4$ & $(7.982, 10^{-5})$ & $(3.536, 0)$ & $(3.602, 0)$ \\
			\hline
		\end{tabular}
		\label{tab:privacy}
	\end{table}
	
    \rev{From Table ~\ref{tab:privacy}, for the same noise variance, we can see that the budget for the Huber mechanism is very close to the Laplace mechanism, especially for higher values of variance. Even with the relaxed definition of privacy, the Gaussian mechanism demands a huge privacy budget, making it a less preferable choice for applications with stricter budgets. The Huber mechanism provides a middle ground between the Laplacian and Gaussian mechanisms w.r.t accuracy and privacy.}
	
	We will now proceed to the application of the Huber mechanism to achieve differentially private LRMC. 
	
	\section{Differentially private LRMC using Huber mechanism}
	Having introduced the Huber mechanism for differential privacy, we demonstrate its utility in the area of matrix completion. We aim to introduce privacy to the iterative optimization process, thereby resulting in a differentially private LRMC algorithm. The choice of privacy mechanism is critical, especially in an algorithm involving iterative operations. 
	Huber noise combines the advantage of Gaussian noise that works well with optimization procedures as most of them are tuned for a normal spread of inherent noise in the data, with the low privacy budget of Laplacian noise to give us a new desirable noise hybrid. 
	
	\subsection{Alternating least squares}
	Let $\mathbf{X} \in \mathbb{R}^{m \times n}_{}$ be the sensitive data matrix that is partly filled. 
	Let $\Omega\subseteq\mathbb{N}_m^{}\times\mathbb{N}_n^{}$ be the set indices in which $\mathbf{X}$ has complete entries. \rev{$\mathcal{P}_{\Omega}^{}(\cdot)$ denotes the mask function; if $(i,j)\in\Omega$, $[\mathcal{P}_{\Omega}^{}(\mathbf{X})]_{i,j}^{}=x_{ij}^{}$, otherwise $[\mathcal{P}_{\Omega}^{}(\mathbf{X})]_{i,j}^{}=0$.} 
	As suggested by \cite{candes2010matrix}, to estimate the incomplete entries in $\mathbf{X},$ we approximate it to a low-rank matrix, $\mathbf{Z}$ and minimize over all such matrices. 
	\begin{equation*}
	\min_{\mathbf{Z}}\, \mathcal{L}(\mathcal{P}_{\Omega}^{}(\mathbf{X}-\mathbf{Z}))+\lambda\norm{\mathbf{Z}}_{*}^{}.
	\end{equation*}
	Here, $\mathcal{L}(\cdot)$ is the data fidelity loss function. This optimization ensures that the existing/complete entries of the matrix $\mathbf{X}$ and the low-rank approximation $\mathbf{Z}$ are similar. 
	
	Despite being convex, the nuclear norm makes the problem difficult to solve. To reduce the computational burden, $\mathbf{Z}$ is bilinearly factorized as $\mathbf{U}\mathbf{V}^{\top}_{}$ where $\mathbf{U}\in\mathbb{R}^{m \times r}_{}$ and $\mathbf{V}\in\mathbb{R}^{n \times r}_{}$, $r\ll \min(m,n)$ and the nuclear norm regularization is reduced to Frobenius-norm regularization in $\mathbf{U}$ and $\mathbf{V}$. 
	\begin{equation*}
	\min_{\mathbf{U},\,\mathbf{V}}\, \mathcal{L}(\mathcal{P}_{\Omega}^{}(\mathbf{X}-\mathbf{U}\mathbf{V}^{\top}_{}))+\frac{\lambda}{2}\left(\norm{\mathbf{U}}_{\texttt{F}}^{2}+\norm{\mathbf{U}}_{\texttt{F}}^{2}\right).
	\end{equation*}
	Though this problem is not convex, it is biconvex and typically solved using Alternating minimization. When $\mathcal{L}(\cdot)=\norm{\cdot}_{\texttt{F}}^2,$ the alternating minimization corresponds to Alternating Least Squares (ALS) (\cite{hastie2015matrix}), which can be decomposed over each row $\mathbf{u}_i^{}$ of $\mathbf{U}$ (and over each row $\mathbf{v}_i^{}$ of $\mathbf{V}$). ALS is found to be faster than the nuclear norm regularized problem.  The minimization can be described as 
	\begin{equation*}
	\min_{\mathbf{U},\,\mathbf{V}}\, \norm{\mathcal{P}_{\Omega}^{}(\mathbf{X}-\mathbf{U}\mathbf{V}^{\top}_{})}^2_{\texttt{F}} + \lambda \left(\norm{\mathbf{U}}^2_{\texttt{F}} + \norm{\mathbf{V}}^2_{\texttt{F}}\right).
	\end{equation*}
	For a given $\mathbf{U}$, $\widehat{\mathbf{v}}_j^{}$ can be obtained through the following optimization problem. 
\begin{equation*}
\begin{gathered}
	\min_{\mathbf{v}_j^{}}\,\frac{1}{2}\,\big\|\mathcal{P}_{\widetilde{\Omega}_j^{}}^{}(\widetilde{\mathbf{x}}_j^{})
		-\mathbf{U}_{\widetilde{\Omega}_j^{},\,\mathbb{N}_r^{}}^{}\mathbf{v}_j^{}\big\|_2^{2}
		+\frac{\lambda}{2}\norm{\mathbf{v}_j^{}}^2_{2},
\end{gathered}
\end{equation*}	
which results in the estimate
\begin{equation*}
\begin{gathered}
	\widehat{\mathbf{v}}_j^{} \gets \Big(\big({\mathbf{U}}_{\widetilde{\Omega}_j^{},\,\mathbb{N}_r^{}}^{}\big)^{\top}_{}{\mathbf{U}}_{\widetilde{\Omega}_j^{},\,\mathbb{N}_r^{}}^{}+\lambda\mathbf{I}_r\Big)^{-1}_{} \big({\mathbf{U}}_{\widetilde{\Omega}_j^{},\,\mathbb{N}_r^{}}^{}\big)^{\top}_{} \mathcal{P}_{\widetilde{\Omega}_j^{}}^{}(\widetilde{\mathbf{x}}_j^{}).
\end{gathered}
\end{equation*}	
Similarly, for a given $\mathbf{V}$, $\mathbf{u}_i^{}$ can be estimated through least squares. Thus, ALS updates the estimates $\widehat{\mathbf{U}}$ and $\widehat{\mathbf{V}}$ in an alternating fashion and is summarized in Algorithm \ref{alg:als}. 

Noise is also added to the ALS procedure for preserving privacy similar to \cite{jain2018differentially} and \cite{chien2021private}, as \textbf{it was observed during experimentation that adding noise to the output directly in each iteration did not allow the algorithm to converge leading to poor precision}. In the simulations section, we experiment with ALS by adding all three noise types and going over their pros and cons.

        \begin{algorithm}[ht] \label{alg:als}
		\KwInput{Incomplete data matrix $\mathbf{X},$ regularization parameter $\lambda,$ number of iterations $T_{ALS}^{}$.}
		\KwOutput{The completed matrix $\widehat{\mathbf{Z}}$}
		\vspace{1.5ex}Initialize $\widehat{\mathbf{U}}$ and $\widehat{\mathbf{V}}$ with random entries\\[1.5ex] 
		\For{ $T_{ALS}^{}$ iterations}{
			\For{$i\in\mathbb{N}_m^{}$}
		    {
			$\widehat{\mathbf{u}}_i^{} \gets \Big(\big(\widehat{\mathbf{V}}_{\Omega_i^{},\,\mathbb{N}_r^{}}^{}\big)^{\top}_{}\widehat{\mathbf{V}}_{\Omega_i^{},\,\mathbb{N}_r^{}}^{}+\lambda\mathbf{I}_r\Big)^{-1}_{} \big(\widehat{\mathbf{V}}_{\Omega_i^{},\,\mathbb{N}_r^{}}^{}\big)^{\top}_{} \mathcal{P}_{\Omega_i^{}}^{}(\mathbf{x}_i^{})$ 
		}
			\For{$j\in\mathbb{N}_n^{}$}{
			Generate noise vector $\mathbf{t}$\\
			
			$\widehat{\mathbf{v}}_j^{} \gets \Big(\big(\widehat{\mathbf{U}}_{\widetilde{\Omega}_j^{},\,\mathbb{N}_r^{}}^{}\big)^{\top}_{}\widehat{\mathbf{U}}_{\widetilde{\Omega}_j^{},\,\mathbb{N}_r^{}}^{}+\lambda\mathbf{I}_r\Big)^{-1}_{} \left( \big(\widehat{\mathbf{U}}_{\widetilde{\Omega}_j^{},\,\mathbb{N}_r^{}}^{}\big)^{\top}_{} \mathcal{P}_{\widetilde{\Omega}_j^{}}^{}(\widetilde{\mathbf{x}}_j^{}) + \mathbf{t}\right)$ 
		}
	}
		\textbf{return} {$\widehat{\mathbf{Z}}=\widehat{\mathbf{U}}\widehat{\mathbf{V}}^{\top}_{}$}
		\caption{Noisy Alternating Least Squares (ALS).}
	\end{algorithm}

    \subsection{IRLS for Huber mechanism}

    In the Huber mechanism, we add Huber noise which has a heavier tail compared to the Gaussian. The Huber loss has been used in robust statistics to make the estimator less sensitive to large deviations. We now use the Huber loss as the data fidelity loss i.e., $\mathcal{L}(\cdot)$ to be chosen as ${\rho}_{\alpha}^{}(\cdot)$ as defined in \eqref{eqn:huber_loss}. 
	Thus,
	\begin{equation}\label{eqn:opt}
	\min_{\mathbf{U},\,\mathbf{V}}\, {\rho}_{\alpha}^{}\!\left(\mathcal{P}_{\Omega}^{}(\mathbf{X}-\mathbf{U}\mathbf{V}^{\top}_{})\right) + \lambda \left(\norm{\mathbf{U}}^2_{\texttt{F}} + \norm{\mathbf{V}}^2_{\texttt{F}}\right).
	\end{equation}
	Similar to ALS, we solve \eqref{eqn:opt} through alternating minimization which can be decomposed over each row of $\mathbf{U}$ and $\mathbf{V}$ to avoid complex matrix computations. \rev{Each of the alternating minimization steps involves a Huber loss minimization with a Frobenius norm regularization term,
	\begin{subequations}
	\begin{align}
	    \widehat{\mathbf{U}} \gets \underset{{\mathbf{U}}}{\operatorname{argmin}}\, {\rho}_{\alpha}^{}\!\left(\mathcal{P}_{\Omega}^{}(\mathbf{X}-\mathbf{U}\widehat{\mathbf{V}}^{\top}_{})\right) + \lambda \norm{\mathbf{U}}^2_{\texttt{F}},\label{eq:irls_a}
	    \\
	    \widehat{\mathbf{V}} \gets \underset{{\mathbf{V}}}{\operatorname{argmin}}\,
	    {\rho}_{\alpha}^{}\!\left(\mathcal{P}_{\Omega}^{}(\mathbf{X}-\widehat{\mathbf{U}}\mathbf{V}^{\top}_{})\right) + \lambda \norm{\mathbf{V}}^2_{\texttt{F}}.\label{eq:irls_b}
	\end{align}
	\end{subequations}
	In robust statistics, Iterative Re-weighted Least Squares (IRLS) (\cite{holland1977robust}) is used whenever the optimization involves Huber loss. IRLS is an iterative estimator which is regaining popularity recently (\cite{kummerle2021iteratively}). The works \cite{dollinger1991influence,kalyani2007mse} show that the IRLS procedure approaches the maximum likelihood estimator for Huber loss under certain conditions.	We solve \eqref{eq:irls_a} and \eqref{eq:irls_b} through IRLS, which we term as Regularized Iterative Re-weighted Least Squares (R-IRLS)}
	
	To derive the step-wise iterations for R-IRLS, we equate the gradient of the loss function with respect to $\mathbf{u}_i^{}$ to $\mathbf{0},$  $-\mathbf{V}\mathbf{W}_i^{}(\mathbf{x}_i^{}-\mathbf{V} \mathbf{u}_i^{})+\lambda\,\mathbf{u}_i^{}=\mathbf{0},$
	where 
	\begin{align*}
	\mathbf{W}_i^{} &=\operatorname{{Diag}}\left({w}_{i1}^{},\,{w}_{i2}^{},\,\ldots,\,{w}_{in}^{}\right),\ \quad {w}_{ij}^{}=\frac{\psi_{\alpha}^{}(x_{ij}^{}-\mathbf{v}_j^{\top} \mathbf{u}_i^{})} {x_{ij}^{}-\mathbf{v}_j^{\top} \mathbf{u}_i^{}}
	\end{align*}
	%
	%
	and $\psi_{\alpha}^{}(\cdot)$ is the {Huber influence function} defined in \eqref{eqn:huber_inf}. 
	This is solved through the IRLS algorithm, proposed by \cite{holland1977robust} where $\mathbf{u}_i^{}$ is computed iteratively with intermediate estimates $ \widehat{\mathbf{u}}_i^{}$ and weights $\mathbf{W}_i^{}$. Similarly, rows of $\mathbf{V}$ can also be computed using R-IRLS when $\mathbf{U}$ is fixed. Thus, \eqref{eqn:opt} is solved using Alternating Minimization:\\
	
	\begin{equation} \label{eqn:irls}
	\begin{gathered}
	\mathbf{u}_i^{}\gets \operatorname{R-IRLS}\big(\mathcal{P}_{\Omega_i^{}}^{}(\mathbf{x}_i^{}),\, \mathbf{V}_{\Omega_i^{},\,\mathbb{N}_r^{}}^{},\, \lambda\big)
	\quad \forall i\in\mathbb{N}_m^{}\\
	\mathbf{v}_j^{}\gets \operatorname{R-IRLS}\big(\mathcal{P}_{\widetilde{\Omega}_j^{}}^{}(\widetilde{\mathbf{x}}_j^{}),\, \mathbf{U}_{\widetilde{\Omega}_j^{},\,\mathbb{N}_r^{}}^{},\, \lambda\big)
	\quad \forall j\in\mathbb{N}_n^{}    
	\end{gathered}
	\end{equation}
	
	
	Recall that the $i$-th row of $\mathbf{X}$ is denoted by $\mathbf{x}_i^{},$ which corresponds to the data of the $i$-th user and $\widetilde{\mathbf{x}}_i^{}$ denotes the $i$-th column. As we are interested in preserving the privacy of the user, we choose to add noise to that particular R-IRLS iteration which deals with the rows of the data matrix instead of the columns. This ensures that privacy is preserved over the item embeddings $\mathbf{V}$. Noise could be added to both sets of iterations; but since there is no significant privacy gained by adding noise to the columns of $\mathbf{X}$, we instead opt for noise addition to just one of the variables. \rev{Thus, no noise is added while updating $\widehat{\mathbf{u}}_i^{}$  and hence, we use least squares instead of R-IRLS to update $\widehat{\mathbf{u}}_i^{}$.} 
	The iterative steps for the proposed method are given in Algorithm \ref{alg:main}, which utilizes the function described in Algorithm \ref{alg:rirls}. 
	%
	
	\begin{algorithm}[ht] \label{alg:main}
		\KwInput{Incomplete data matrix $\mathbf{X},$ assumed rank $\mathbf{r}$, Huber transition parameter $\alpha,$ regularization parameter $\lambda,$ number of iterations $N$.}
		\KwOutput{The completed matrix $\widehat{\mathbf{Z}}$ }
		
		Initialize $\widehat{\mathbf{U}}$ and $\widehat{\mathbf{V}}$ with random entries\\
		\For{ $N$ iterations}{
			\For{$i\in\mathbb{N}_m^{}$}
			{
				$\widehat{\mathbf{u}}_i^{} \gets \Big(\big(\widehat{\mathbf{V}}_{\Omega_i^{},\,\mathbb{N}_r^{}}^{}\big)^{\top}_{}\widehat{\mathbf{V}}_{\Omega_i^{},\,\mathbb{N}_r^{}}^{}+\lambda\mathbf{I}_r\Big)^{-1}_{} \big(\widehat{\mathbf{V}}_{\Omega_i^{},\,\mathbb{N}_r^{}}^{}\big)^{\top}_{} \mathcal{P}_{\Omega_i^{}}^{}(\mathbf{x}_i^{})$ 
			}
			\For{$j\in\mathbb{N}_n^{}$}{
		$\widehat{\mathbf{v}}_j^{}\gets \operatorname{R-IRLS}\big(\mathcal{P}_{\widetilde{\Omega}_j^{}}^{}(\widetilde{\mathbf{x}}_j^{}),\, \widehat{\mathbf{U}}_{\widetilde{\Omega}_j^{},\,\mathbb{N}_r^{}}^{},\, \mathbf{r}, \, \alpha,\, \lambda\big)$ \hfill $\triangleright$ Update estimate 
			}
		}
		\textbf{return} {$\widehat{\mathbf{Z}}=\widehat{\mathbf{U}}\widehat{\mathbf{V}}^{\top}_{}$}
		\caption{IRLS+Huber}
	\end{algorithm} 
	
	\begin{algorithm}[ht] \label{alg:rirls}
		\KwInput{Targets $\mathbf{y},$ data matrix $\mathbf{A},$ assumed rank $\mathbf{r},$ Huber transition parameter $\alpha,$ regularization parameter $\lambda,$ number of iterations $K$.}
		\KwOutput{The IRLS estimate $\widehat{\boldsymbol{\theta}}$}
		\vspace{1.5ex}Initialize $\widehat{\boldsymbol{\theta}} \sim \mathcal{N}(0,1)$\\[1.5ex] 
		\For{ $K$ iterations}{
			$\mathbf{W} \gets \operatorname{{Diag}}\left(
				\frac{\psi_{\alpha}^{}\left(y_{1}^{}-\mathbf{a}_1^{\top} \widehat{\boldsymbol{\theta}}\right)} {y_{1}^{}-\mathbf{a}_1^{\top} \widehat{\boldsymbol{\theta}}},\,
				\ldots,\,
				\frac{\psi_{\alpha}^{}\left(y_{p}^{}-\mathbf{a}_p^{\top} \widehat{\boldsymbol{\theta}}\right)} {y_{p}^{}-\mathbf{a}_p^{\top} \widehat{\boldsymbol{\theta}}}\right)$ \hfill $\triangleright$ Update weights\\
				Generate noise vector $\mathbf{t}$ 
				\\
				$\widehat{\boldsymbol{\theta}} \gets \left(\mathbf{A}^{\top}_{}\mathbf{W}\mathbf{A}+\lambda \mathbf{I}_q^{}\right)^{-1}_{} \left(\mathbf{A}^{\top}_{}\mathbf{W}\mathbf{y} + \mathbf{t}\right)$ \hfill $\triangleright$ Update estimate 
	    }
		\textbf{return} {$\widehat{\boldsymbol{\theta}}$}
		\caption{$\operatorname{R-IRLS}\big(\mathbf{y},\, \mathbf{A},\, \mathbf{r}, \, \alpha,\, \lambda)$}
	\end{algorithm}

The IRLS + Laplacian method swaps out the Huber noise for Laplacian noise drawn from the distribution $\mathcal{L}(\Delta f/\epsilon)$. We also experiment with noisy IRLS in the simulation section and examine its comparative performance with noisy ALS.

	\section{Simulation Results}
	We present the empirical results for the Huber mechanism of noise addition for the problem of low-rank matrix completion 
	and compare it with the Gaussian and Laplace mechanisms optimized using the ALS and IRLS procedure.  We consider 3 different datasets: 
	\begin{enumerate}
	    \item A synthetic dataset in which we generate matrices of a specific rank as $\mathbf{X}=\mathbf{U}\mathbf{V}^{\top}_{}$ with $\mathbf{U}$ and $\mathbf{V}$ generated randomly. A percentage of the entries are sampled at random and replaced with zeros to indicate the incomplete entries.
	    \item The MovieLens100k dataset (\cite{movielens}). The dataset consists of approximately 1000 users and 100000 entries. The percentage of observable entries is $\sim 5\%$.
	    \item The Sweet Recommender System dataset (\cite{sweet_rs}) which contains 2000 users and over 45000 ratings ($\sim 40\%$ of entries are observable). 
	\end{enumerate}

	\begin{table}[ht]
		\centering
		\caption{Simulation parameters}
		\begin{tabular}{|c|c|}
			\hline
			Parameter & Value \\
			\hline
			Number of alternating steps $T$ & 50\\
			IRLS iterations $K$ & 20\\
			Regularization parameter $\lambda$ & 0.5\\
			Trials for averaging & 10\\
			Ratings range & 1 - 5\\
			\hline
		\end{tabular}
		\label{tab:sim_par}
	\end{table}

\begin{table}[!ht]
    \centering
    \caption{RMSE for synthetic dataset of rank \textbf{5}}
    \begin{tabular}{|c|c|c|c|c|c|c|}
    \hline
        \textbf{Variance} & \textbf{Observed fraction} & \textbf{Algorithm} & \textbf{Vanilla} & \textbf{Gaussian} & \textbf{Laplacian} & \textbf{Huber} \\ \hline
        \multirow{6}{*}{1} & \multirow{2}{*}{5\%} & ALS & 0.2040 & 0.3740 & 0.3701 & \textbf{0.3569} \\ \cline{3-7}
        &  & IRLS & 0.2039 & 0.3728 & 0.3939 & \textbf{0.3640} \\ \cline{2-7}
        & \multirow{2}{*}{10\%} & ALS & 0.0948 & 0.2008 & 0.2017 & \textbf{0.2002} \\ \cline{3-7}
         &  & IRLS & 0.0950 & 0.2000 & 0.2013 & \textbf{0.1995} \\ \cline{2-7}
         & \multirow{2}{*}{15\%} & ALS & 0.0567 & 0.1924 & 0.1926 & \textbf{0.1923} \\ \cline{3-7}
         &  & IRLS & 0.0567 & 0.1927 & 0.1924 & \textbf{0.1901} \\ \hline
        \multirow{6}{*}{2} & \multirow{2}{*}{5\%} & ALS & 0.2481 & 0.6388 & 0.6338 & \textbf{0.6261} \\ \cline{3-7}
         &  & IRLS & 0.2482 & 0.6370 & 0.6427 & \textbf{0.6194} \\ \cline{2-7}
         & \multirow{2}{*}{10\%} & ALS & 0.0981 & 0.2203 & 0.2209 & \textbf{0.2111} \\ \cline{3-7}
         &  & IRLS & 0.0982 & 0.2194 & 0.2208 & \textbf{0.2047} \\ \cline{2-7}
         & \multirow{2}{*}{15\%} & ALS & 0.0550 & 0.2018 & 0.2022 & \textbf{0.1950} \\ \cline{3-7}
         &  & IRLS & 0.0551 & 0.2016 & 0.2019 & \textbf{0.1947} \\ \hline
    \end{tabular}
    \label{tab:res_syn_5}
\end{table}
	The simulation parameters are tabulated in Table \ref{tab:sim_par}. \rev{The metric used for measuring the performance of the algorithms is Root Mean Squared Error (RMSE), which is determined as $ \norm{\mathbf{X}-\mathbf{U}\mathbf{V}^{\top}_{}}_\mathtt{F}/\sqrt{mn}$ (as considered in \cite{chien2021private}, \cite{liu2013ftf} and \cite{liu2015}).}
	
	For synthetic data, we consider several cases of noise variance, fraction of observed entries and rank of $\mathbf{X}$ to evaluate the performance of the various mechanisms using both the ALS and IRLS procedures to decide the best course of action to collect and compare results from real datasets. 
	Please note that the variance of Huber noise approaches 1 asymptotically as $\alpha \rightarrow \infty$. But, it decays rapidly and is approximately $1$ even at $\alpha = 3$, which is employed to generate Huber noise of unit variance.
	
		\paragraph{Synthetic datasets} We tabulate RMSE for synthetic data with various parameters (such as the variance of additive noise, percentage of observable entries and optimization procedure) in Table \ref{tab:res_syn_5}. The lowest RMSE in each case is provided in \textbf{bold}. From Table \ref{tab:res_syn_5}, there are two key observations. \rev{Firstly, Huber noise gives consistently better results across different variances and visible fractions. Recall from Table \ref{tab:privacy} (in page 4) that the Huber mechanism provides similar privacy to the Laplace mechanism for the same variance (for higher variance values). However, it seems to provide better accuracy for the same amount of noise being added as evidenced by Table ~\ref{tab:res_syn_5}. The Gaussian mechanism provides much lower privacy (greater value of $\epsilon$) but does not significantly reduce RMSE over the Huber and Laplace mechanisms. 
	Secondly, IRLS gives significant improvement over ALS with regards to Huber noise, especially for data with lesser observable fraction. This could be because IRLS is able to handle large deviations brought about by the Huber noise unlike ALS, without affecting privacy.}

It is also observed that an increase in the variance of added noise leads to an overall decrease in accuracy, as expected. However, the effect is heavily pronounced for data with a lesser fraction of observed entries. Moreover, this behaviour is more evident for Gaussian and Laplacian noise when compared to Huber noise which makes the Huber mechanism more preferable for adding noise of higher variance.

In the case of Huber noise, we observe that ALS gives a better performance than IRLS for lower variance and vice-versa for higher variances. This may occur due to the fact that the Huber and Gaussian distributions are quite similar at lower variances (larger $\alpha$).
\rev{In line with theoretical expectations, ALS gives better performance consistently for Gaussian noise. However, in most cases, IRLS gives better accuracy than ALS for Laplacian noise making the former more suitable for the Laplace mechanism. We also analyze the variation in performance with rank of the matrix $\mathbf{X}$; these results are provided in Appendix ~\ref{sec:add_syn}. We notice that with an increase in rank, Gaussian mechanism provides the lowest RMSE in a few cases although Huber mechanism gives the best performance overall. However, it is worth noting that the privacy guarantees for the Gaussian mechanism are much weaker than both Huber and Laplacian.}

\paragraph{Real datasets}
Based on the results from Synthetic data, we compare the results of ALS + Gaussian noise, IRLS + Laplacian noise and IRLS + Huber noise for the MovieLens dataset and SweetRS dataset in Tables \ref{tab:res_mov} and \ref{tab:res_sweet} respectively. No noise is added in the cases of baseline ALS and baseline IRLS and they are hence not private. 
Rank is set to be $32$ for both the datasets. For Movielens, $20$ iterations of ALS are performed and for SweetRS $T_{ALS}=100$. For SweetRS dataset, $40\%$ of ratings are available. However, for experiments where a lower number of observable entries are required, the entries are sub-sampled to produce the desired fraction.

\rev{For synthetic datasets, we provide comparisons across noise addition mechanisms for the same optimization procedure. Here, as we are comparing across optimization procedures as well, We note that performing 20 iterations of IRLS requires much higher computation time when compared to ALS. Therefore, we also provide results for a variation of IRLS that perform only 2 IRLS iterations. A key observation is that there is no significant increase in the RMSE when we limit to just two iterations of IRLS. Thus, reasonable accuracy is achieved with less complexity.}

\begin{table}[h]
		\centering
		\caption{RMSE for MovieLens100k, Visible fraction = 5.1\%}
		\begin{tabular}{| c | c | c | c | c | c | c |}
			\hline
			& \textbf{ALS}
			&\textbf{IRLS}
			&  \textbf{ALS + G} & \textbf{IRLS + L}
			& \textbf{IRLS + H}
			& \textbf{IRLS-2 + H}\\ \hline
			\textbf{MSE} & $1.2463$ & $1.2787$ & $1.3883$ & $1.3952$ & $\textbf{1.3755}$ & $1.3850$
			\\
			\hline
		\end{tabular}
		\label{tab:res_mov}
	\end{table}
\begin{table}[ht]
		\centering
		\caption{RMSE for SweetRS}
		\label{tab:res_sweet}
		\begin{tabular}{| c | c | c | c | c | c | c |}
			\hline
			\textbf{Visible fraction} & \textbf{ALS}
			&\textbf{IRLS}
			&  \textbf{ALS + G} & \textbf{IRLS + L}
			& \textbf{IRLS + H} & \textbf{IRLS-2 + H}
			\\ \hline
			$5\%$ & $2.1333$ & $2.1334$ & $2.2103$ & $2.2096$ & $\textbf{2.2022}$ & $2.2137$\\ \hline
			$10\%$ & $1.7103$ & $1.7102$ & $1.8830$ & $1.8827$ & $\textbf{1.8757}$ & $1.8820$ \\ \hline
			$15\%$ & $1.6881$ & $1.6891$ & $\textbf{1.7686}$ & $1.7756$ & $1.7694$ & $1.7791$
			\\
			\hline
		\end{tabular}
	\end{table}

\rev{Similar to the trend observed for synthetic data, we note that Huber mechanism results in the lowest RMSE for both MovieLens100k as well and SweetRS datasets when $5\%$ of the entries are visible. As the percentage of observable entries increases, Gaussian mechanism takes over. As mentioned earlier, Gaussian mechanism provides weaker privacy guarantees.}

	\section{Discussions}
	 
	We can observe that Laplacian noise gives satisfactorily accurate results for the low privacy budget that it requires. But the only drawback is that the large tail of the Laplacian distribution makes the results imprecise if consistent performance is important, and in most cases it is. Gaussian also works well, especially for datasets with a small observable fraction, but its large privacy budget works against its favour. Huber noise gives the best trade-off overall, with the double benefits of both high accuracy and a low privacy budget.
	
	Although Gaussian noise is preferred to achieve differentially private ALS, we observe from our simulation results that other noise mechanisms perform competitively while offering $\epsilon$-DP guarantee as compared to the $(\epsilon,\delta)$-DP guarantee provided by Gaussian. Therefore, we conclude that the choice of noise mechanism is not obvious and exploring other mechanisms is crucial in improving performance. 
	
	Though we have explored the efficacy of the Huber mechanism for differentially private matrix completion, we believe that the mechanism holds its own merits as a noise addition mechanism for differential privacy. \rev{While constructing private algorithms, typically there is a trade-off between accuracy and privacy. In this work, we have provided exact privacy guarantees of the Huber mechanism and empirically compared its accuracy with other noise mechanisms in the context of matrix completion. This is because the exact characterization of the accuracy of LRMC algorithms is difficult. It is worth exploring the merits of the Huber mechanism analytically and in the general context of differential privacy.}
	\medskip
	\bibliographystyle{unsrtnat}
	\bibliography{ref.bib}
	

\newpage
\appendix

\section{Appendix: Derivation of privacy for Huber mechanism} \label{sec:A1}
\paragraph{Case 1: $\Delta f \leq 2\alpha$} 
We go over all the intervals of $x$ in order and find the value of the function $g(x)$.

\begin{enumerate} \renewcommand{\labelenumi}{\roman{enumi}. }
    
    \item For $x < -\Delta f - \alpha$, $\abs{x}>\alpha$ and $\abs{x + \Delta f}>\alpha$ and both $x$ and $x + \Delta f$ are negative.
    \begin{equation*}
    \begin{split}
        g_1(x) 
        = \alpha\left(-x - \Delta f - \frac{\alpha}{2}\right) - \alpha\left(-x- \frac{\alpha}{2}\right)
        = -\alpha\,\Delta f.
    \end{split}
    \end{equation*}
    Since $g_1(x)$ is constant, the maximum of $g_1(x)$ in this range, $g_{1,{max}}=-\alpha\,\Delta f$.
    
    \item For $-\Delta f - \alpha \leq x \leq - \alpha$, $\abs{x}\geq\alpha$ whereas $\abs{x + \Delta f} \leq \alpha$.
    \begin{equation*}
    \begin{split}
        g_2(x) 
        = \frac{(x + \Delta f)^2}{2} - \alpha\left(-{x} - \frac{\alpha}{2}\right)
        = \frac{(x + \alpha + \Delta f)^2}{2} - \alpha\,\Delta f.
    \end{split}
    \end{equation*}
    The function is then monotonically increasing. Thus, 
    $g_{2,{max}}$, occurs at $x = -\alpha$, 
        $g_{2,{max}} = g_2(x) \big\rvert_{x = -\alpha} 
                = {\Delta f(\Delta f - 2\alpha)}/{2}
                \leq 0$
    , since $\Delta f \leq 2\alpha$ by the case definition.
    
    \item For $-\alpha < x \leq \alpha - \Delta f$, $\abs{x}\leq\alpha$ and $\abs{x + \Delta f}\leq\alpha$. Hence ,
    \begin{equation*}
    \begin{split}
        g_3(x) &= \frac{(x + \Delta f)^2}{2} - \frac{x^2}{2} = x\Delta f + \frac{\Delta f^2}{2} .
    \end{split}
    \end{equation*}
    As $g_3(x)$ is monotonically increasing, 
        $g_{3,{max}} = g_3(x) \big\rvert_{x = \alpha - \Delta f}
                = \alpha\,\Delta f - \frac{\Delta f^2}{2}
                \leq \alpha\,\Delta f.$
    
    \item For $\alpha - \Delta f < x \leq \alpha$, $\abs{x}\leq \alpha$ whereas $\abs{x + \Delta f} >\alpha$. Now, $x + \Delta f>0$. So, 
    \begin{equation*}
    \begin{split}
        g_4(x) 
        = \alpha\left(x + \Delta f - \frac{\alpha}{2}\right) - \frac{x^2}{2}
            = \alpha\,\Delta f - \frac{1}{2}(x-\alpha)^2.
    \end{split}
    \end{equation*}
    The maximum value in this range occurs at $x = \alpha$. 
    Hence, 
        $g_{4,{max}} = \alpha\,\Delta f.$
    \item For $x > \alpha$, both $\abs{x}>a$ and $\abs{x + \Delta f}>\alpha$, both $x$ and $x + \Delta f$ are positive.
    \begin{equation*}
    \begin{split}
        g_5(x) 
        = \alpha\left(x + \Delta f - \frac{\alpha}{2}\right) - \alpha\left(x - \frac{\alpha}{2}\right) = \alpha\,\Delta f.
    \end{split}
    \end{equation*}
    Since the value of $g_5(x)$ comes out to be a constant, 
    $g_{5,{max}}=\alpha\,\Delta f$. \\
\end{enumerate}
Thus, the overall upper bound $g_{max}$ for Case 1 can be computed as
\begin{equation*}
    \begin{split}
        g_{max} = \max_{i=1,\ldots,5}\, g_{i,{max}} 
        = \alpha\,\Delta f
    \end{split}
    \end{equation*}

\paragraph{ Case 2: $\Delta f > 2\alpha$} 
Similar to case 1, we compute the values for $g(x)$ for different intervals of $x$.

\begin{enumerate} \renewcommand{\labelenumi}{\roman{enumi}. }
    
    \item For $x < -\Delta f - \alpha$, $g_1(x)$ and $g_{1,{max}}$ are identical to those in Case 1.
    
    \item For $-\Delta f - \alpha \leq x \leq \alpha - \Delta f$, $\abs{x}$ remains greater than $\alpha$ whereas $\abs{x + \Delta f} \leq \alpha$. 
    \begin{equation*}
    \begin{split}
         g_2(x) &= \frac{(x + \Delta f)^2}{2} - \alpha\left(\abs{x} - \frac{\alpha}{2}\right)
         = \frac{(x + \alpha + \Delta f)^2}{2} - \alpha  \Delta f.
    \end{split}
    \end{equation*}
    The maximum value of $g_2(x)$ occurs at $x = \alpha - \Delta f$. 
    Hence, 
        $g_{2,{max}} = g_2(x) \big\rvert_{x = \alpha - \Delta f}
        = \alpha(2\alpha - \Delta f)
        \leq 0$
    , since $\Delta f > 2\alpha$ by the case definition.
    
    \item For $\alpha - \Delta f < x \leq -\alpha$, $\abs{x}\geq \alpha$ and $\abs{x + \Delta f} > \alpha$. Also, note that $x<0$ but $x + \Delta f>0$.
    \begin{equation*}
    \begin{split}
         g_3(x) 
         = \alpha\left(x + \Delta f - \frac{\alpha}{2}\right) - \alpha\left(-x - \frac{\alpha}{2}\right)
         = \alpha(2x + \Delta f).
    \end{split}
    \end{equation*}
    The maximum value 
    occurs at $x = -\alpha$, 
    $    g_{3,{max}} = g_3(x) \big\rvert_{x = -\alpha}
        = \alpha\,\Delta f - 2\alpha^2
        \leq \alpha\,\Delta f
    $
    .
    
    \item For $-\alpha < x \leq \alpha$, $\abs{x} \leq \alpha$. So,
    \begin{equation*}
    \begin{split}
        g_4(x) &= \alpha\left(\abs{x + \Delta f} - \frac{\alpha}{2}\right) - \frac{x^2}{2}
        = \alpha\left(x + \Delta f - \frac{\alpha}{2}\right) - \frac{x^2}{2}
        = \alpha\,\Delta f - \frac{1}{2}(x-\alpha)^2.
    \end{split}
    \end{equation*}
    The maximum value 
    occurs at $x = \alpha$, 
      $g_{4,{max}} = \alpha\,\Delta f.$
    \item For $x > \alpha$ too, $g_5(x)$ and $g_{5,{max}}$ are identical to those in Case 1.
\end{enumerate}
The overall upper bound $g_{max}$ for Case 2 is
\begin{equation*}
    \begin{split}
    g_{max} = \max_{i=1,\ldots,5}\,g_{i,max}
    = \alpha\,\Delta f.
\end{split}
    \end{equation*}
	
\section{Appendix: Additional results for Synthetic data} \label{sec:add_syn}

We present the results for the synthetic data when the rank of $\mathbf{X}$ is set to $10$ and $20$ in Tables ~\ref{tab:res_syn_10} and ~\ref{tab:res_syn_20} respectively. When compared to Table ~\ref{tab:res_syn_5}, these results show an increase in MSE owing to the increased complexity of a higher rank structure. However, the rest of the trends are similar to those observed in the rank $5$ data. 


\begin{table}[!ht]
    \centering
    \caption{Results for Synthetic dataset of rank \textbf{10}}
    \begin{tabular}{|c|c|c|c|c|c|c|}
    \hline
        \textbf{Variance} & \textbf{ Observed fraction} & \textbf{Algorithm} & \textbf{Vanilla} & \textbf{Gaussian} & \textbf{Laplacian}& \textbf{Huber} \\ \hline
        \multirow{6}{*}{1} & \multirow{2}{*}{5\%} & ALS & 0.4058 & 0.6471 & \textbf{0.6452} & 0.6720 \\ \cline{3-7}
         &  & IRLS & 0.4059 & 0.6376 & 0.6888 & \textbf{0.6479} \\ \cline{2-7}
         & \multirow{2}{*}{10\%} & ALS & 0.2317 & \textbf{0.3015} & 0.3029 & 0.3016 \\ \cline{3-7}
         &  & IRLS & 0.2317 & \textbf{0.3023} & 0.3037 & 0.3031 \\ \cline{2-7}
         & \multirow{2}{*}{15\%} & ALS & 0.1602 & 0.2770 & 0.2774 & \textbf{0.2768} \\\cline{3-7}
         & ~ & IRLS & 0.1596 & 0.2773 & 0.2771 & \textbf{0.2767} \\ \hline
        \multirow{6}{*}{2} & \multirow{2}{*}{5\%} & ALS & 0.3971 & 0.6454 & 0.6674 & \textbf{0.6361} \\ \cline{3-7}
         &  & IRLS & 0.3972 & 0.6835 & 0.7347 & \textbf{0.6646} \\ \cline{2-7}
        ~ & \multirow{2}{*}{10\%} & ALS & 0.2288 & \textbf{0.3075} & 0.3086 & 0.3098 \\ \cline{3-7}
        ~ & ~ & IRLS & 0.2317 & 0.3086 & 0.3080 & \textbf{0.3040} \\ \cline{2-7}
        ~ & \multirow{2}{*}{15\%} & ALS & 0.1599 & 0.2836 & 0.2839 & \textbf{0.2802} \\ \cline{3-7}
        ~ & ~ & IRLS & 0.1602 & 0.2822 & 0.2838 & \textbf{0.2799} \\ \hline
    \end{tabular}
    \label{tab:res_syn_10}
\end{table}


\begin{table}[!ht]
    \centering
    \caption{Results for Synthetic dataset of rank \textbf{20}}
    \begin{tabular}{|c|c|c|c|c|c|c|}
    \hline
    \textbf{Variance} & \textbf{Observed fraction} & \textbf{Algorithm} & \textbf{Vanilla} & \textbf{Gaussian} & \textbf{Laplacian}& \textbf{Huber} \\ \hline
        \multirow{6}{*}{1} & \multirow{2}{*}{5\%} & ALS & 0.6902 & 1.5165 & 1.5330 & \textbf{1.5142} \\ \cline{3-7}
        ~ & ~ & IRLS & 0.6886 & 1.5493 & 1.5715 & \textbf{1.5191} \\ \cline{2-7}
        ~ & \multirow{2}{*}{10\%} & ALS & 0.4622 & 0.5087 & 0.5113 & \textbf{0.5083} \\ \cline{3-7}
        ~ & ~ & IRLS & 0.4600 & \textbf{0.5054} & 0.5189 & 0.5083 \\ \cline{2-7}
        ~ & \multirow{2}{*}{15\%} & ALS & 0.3961 & 0.4520 & 0.4529 & \textbf{0.4516} \\ \cline{3-7}
        ~ & ~ & IRLS & 0.3957 & 0.4513 & 0.4526 & \textbf{0.4494} \\ \hline
        \multirow{6}{*}{2} & \multirow{2}{*}{5\%} & ALS & 0.6126 & \textbf{1.8677} & 1.8843 & 1.8770 \\ \cline{3-7}
        ~ & ~ & IRLS & 0.6136 & 1.9546 & 1.9966 & \textbf{1.9524} \\ \cline{2-7}
        ~ & \multirow{2}{*}{10\%} & ALS & 0.4648 & 0.5222 & 0.5248 & \textbf{0.5159} \\ \cline{3-7}
        ~ & ~ & IRLS & 0.4631 & 0.5220 & 0.5258 & \textbf{0.5154} \\ \cline{2-7}
        ~ & \multirow{2}{*}{15\%} & ALS & 0.3895 & 0.4470 & 0.4486 & \textbf{0.4448} \\ \cline{3-7}
        ~ & ~ & IRLS & 0.3960 & 0.4474 & 0.4488 & \textbf{0.4429} \\ \hline
    \end{tabular}
    \label{tab:res_syn_20}
\end{table}

\end{document}